\documentclass[11pt,english,preprintnumbers,amsmath,amssymb,superscriptaddress,nofootinbib,prl]{revtex4}
\usepackage{mathpazo}
\usepackage[T1]{fontenc}
\usepackage[latin9]{inputenc}
\usepackage[letterpaper]{geometry}
\usepackage{babel}
\usepackage{array}
\usepackage{verbatim}
\usepackage{multirow}
\usepackage{amsmath}
\usepackage{amsthm}
\usepackage{amssymb}
\usepackage{graphicx}
\usepackage[unicode=true]
 {hyperref}

\makeatletter

\providecommand{\tabularnewline}{\\}

\@ifundefined{textcolor}{}
{%
 \definecolor{BLACK}{gray}{0}
 \definecolor{WHITE}{gray}{1}
 \definecolor{RED}{rgb}{1,0,0}
 \definecolor{GREEN}{rgb}{0,1,0}
 \definecolor{BLUE}{rgb}{0,0,1}
 \definecolor{CYAN}{cmyk}{1,0,0,0}
 \definecolor{MAGENTA}{cmyk}{0,1,0,0}
 \definecolor{YELLOW}{cmyk}{0,0,1,0}
}
 \theoremstyle{definition}
 \newtheorem*{defn*}{\protect\definitionname}
  \theoremstyle{plain}
  \ifx\thechapter\undefined
    \newtheorem{lem}{\protect\lemmaname}
  \else
    \newtheorem{lem}{\protect\lemmaname}[chapter]
  \fi
  \theoremstyle{definition}
  \ifx\thechapter\undefined
    \newtheorem{defn}{\protect\definitionname}
  \else
    \newtheorem{defn}{\protect\definitionname}[chapter]
  \fi
  \theoremstyle{plain}
  \ifx\thechapter\undefined
    \newtheorem{thm}{\protect\theoremname}
  \else
    \newtheorem{thm}{\protect\theoremname}[chapter]
  \fi

\makeatother

  \providecommand{\definitionname}{Definition}
  \providecommand{\lemmaname}{Lemma}
\providecommand{\theoremname}{Theorem}

\begin{document}

\title{The gap of the area-weighted Motzkin spin chain is exponentially
small}

\author{Lionel Levine}
\email{levine@math.cornell.edu}

\selectlanguage{english}%

\affiliation{Department of Mathematics, Cornell University, Ithaca, NY 14853}

\author{Ramis Movassagh }
\email{q.eigenman@gmail.com}

\selectlanguage{english}%

\affiliation{Department of Mathematical Sciences, IBM T. J . Watson Research Center,
Yorktown Heights, NY 10598}

\date{\today}

\maketitle
We prove that the energy gap of the model proposed by Zhang, Ahmadain,
and Klich \cite{zhang2016quantum} is exponentially small in the square
of the system size. In \cite{MovassaghShor2016} a class of exactly
solvable quantum spin chain models was proposed that have integer
spins ($s$), with a nearest neighbors Hamiltonian, and a unique ground
state. The ground state can be seen as a uniform superposition of
all $s-$colored Motzkin walks. The half-chain entanglement entropy
provably violates the area law by a square root factor in the system's
size ($\sim\sqrt{n}$) for $s>1$. For $s=1$, the violation is logarithmic
\cite{Movassagh2012_brackets}. Moreover in \cite{MovassaghShor2016}
it was proved that the gap vanishes polynomially and is $O(n^{-c})$
with $c\ge2$. 

Recently, a deformation of \cite{MovassaghShor2016}, which we call
``weighted Motzkin quantum spin chain'' was proposed \cite{zhang2016quantum}.
This model has a unique ground state that is a superposition of the
$s-$colored Motzkin walks weighted by $t^{\text{area\{Motzkin walk\}}}$
with $t>1$. The most surprising feature of this model is that it
violates the area law by a factor of $n$. Here we prove that the
gap of this model is upper bounded by $8ns\text{ }t^{-n^{2}/3}$ for
$t>1$.

\tableofcontents{}

\section{\label{sec:Context-and-summary}Context and summary of the results}

In recent years there has been a surge of activities in developing
new exactly solvable models that give large violations of the area
law for the entanglement entropy \cite{Movassagh2012_brackets,MovassaghShor2016,dell2016violation,salberger2016fredkin,zhang2016quantum}.
The notion of exactly solvable in these works means that the ground
state can be written down analytically and the gap to the first excited
state can be quantified. In some cases certain correlation functions
can be analytically calculated as well (e.g., \cite{movassagh2016entanglement}).
Understanding the gap is important for the physics of quantum many-body
systems.

Area law says that the entanglement entropy of the ground state of
a gapped Hamiltonian between a subsystem and the rest scales as the
boundary of the subsystem. This has only rigorously been proved in
one dimension \cite{Matth_areal}, yet is believed to hold in higher
dimensions as well. For gapless one-dimensional systems, based on
detailed and precise results in critical systems described by conformal
field theories, the area law was believed to be violated by at most
a logarithmic factor in the system's size. The above presume physical
reasonability of the underlying model, which means the Hamiltonian
is local, translationally invariant in the bulk with a unique ground
state.

In \cite{MovassaghShor2016} a class of exactly solvable quantum
spin-chains was proposed that violate the area law by a square root
factor in the system's size. They have positive integer spins ($s>1$),
the Hamiltonian is nearest neighbors with a unique ground state that
can be seen as a uniform superposition of $s-$colored Motzkin walks.
The half-chain entanglement entropy provably scales as a square root
factor in the system's size ($\sim\sqrt{n}$). The power-law violation
of the entanglement entropy in that work provides a counter-example
to the widely believed notion, that translationally invariant spin
chains with a unique ground state and local interactions can violate
the area law by at most a logarithmic factor in the system's size. 

This 'super-critical' violation of the area law for a physical system
has inspired follow-up works; most notable are \cite{salberger2016fredkin,dell2016violation}
and \cite{zhang2016quantum}. A class of \textit{half}-integer spin
chains, called Fredkin spin chain \cite{salberger2016fredkin,dell2016violation},
was proposed with similar behavior and scaling of the entanglement
entropy as in \cite{MovassaghShor2016}. 

More recently, a deformation of the Hamiltonian in \cite{MovassaghShor2016}
was proposed by Z. Zhang, A. Ahmadain, I. Klich, in which, the ground
state is a superposition of all Motzkin walks weighted by the area
between the Motzkin walk and the horizontal axis \cite{zhang2016quantum}.
The half-chain entanglement entropy of this model violates the area
law with the maximum possible scaling factor with the system's size
(i.e., $n$). However, they did not quantify the gap to the first
excited state.

In \cite{MovassaghShor2016} the gap to the first excited state was
proved to scale as $n^{-c}$, where $c\ge2$. Later it was shown that
the gap of Fredkin spin chain has the same scaling with the system's
size \cite{movassagh2016gap}. Here we prove an upper bound on the
gap of the weighted-Motzkin quantum spin-chain proposed in \cite{zhang2016quantum}
that scales as $8ns\text{ }t^{-n^{2}/3}$ where $t>1$. 

We remind our reader the asymptotic notations:
\begin{itemize}
\item $g(n)$ is $O(f(n))$ if and only if for some constants $c$ and $n_{0}$,
$g(n)\le cf(n)$ for all $n\ge n_{0}$,
\item $g(n)$ is $\Omega(f(n))$ if for some constants $c$ and $n_{0}$,
$g(n)\ge cf(n)$ for all $n\ge n_{0}$,
\item $g(n)$ is $\Theta(f(n))$ if $g(n)=O(f(n))$ and $g(n)=\Omega(f(n))$. 
\end{itemize}
Let us denote the gap by $\Delta$. Below if we want to emphasize
the gap of a particular Hamiltonian, we write $\Delta(H)$ for the
gap. In the table below, we summarize what is known about the recent
results that achieve ``super-critical'' scaling of the entanglement
entropy in physical quantum spin-chains:\\
\begin{center}
\begin{tabular}{|c|c|c|c|}
\hline 
\multirow{2}{*}{The model} & \multirow{2}{*}{Spin dimension} & Entanglement Entropy & \multirow{2}{*}{Gap}\tabularnewline
 &  & Approximately & \tabularnewline
\hline 
\cite{MovassaghShor2016} & $s>1$ integer & $\sqrt{n}\log(s)$ & $\Theta(n^{-c})\le\Delta\le\Theta(n^{-2}),\quad c\gg1$\tabularnewline
\hline 
\cite{salberger2016fredkin,dell2016violation}  & $s>1/2$ half-integer & $\sqrt{n}\log(s)$ & $\Theta(n^{-c})\le\Delta\le\Theta(n^{-2}),\quad c\gg1$ \cite{movassagh2016gap}\tabularnewline
\hline 
\cite{zhang2016quantum} & $s$ positive integer & $n\log(s)$ & $\Delta\le8ns\text{ }t^{-n^{2}/3},\quad t>1${*}\tabularnewline
\hline 
\multicolumn{4}{|l|}{{*} Proved in this paper.}\tabularnewline
\hline 
\end{tabular}
\par\end{center}
\section{\label{sec:Hamiltonian-Gap-of}Hamiltonian Gap of Recent Exactly
Solvable Models}
\subsection{\label{subsec:Movassagh_shor}The Motzkin quantum spin chain}
\begin{figure}
\centering{}\includegraphics[scale=0.3]{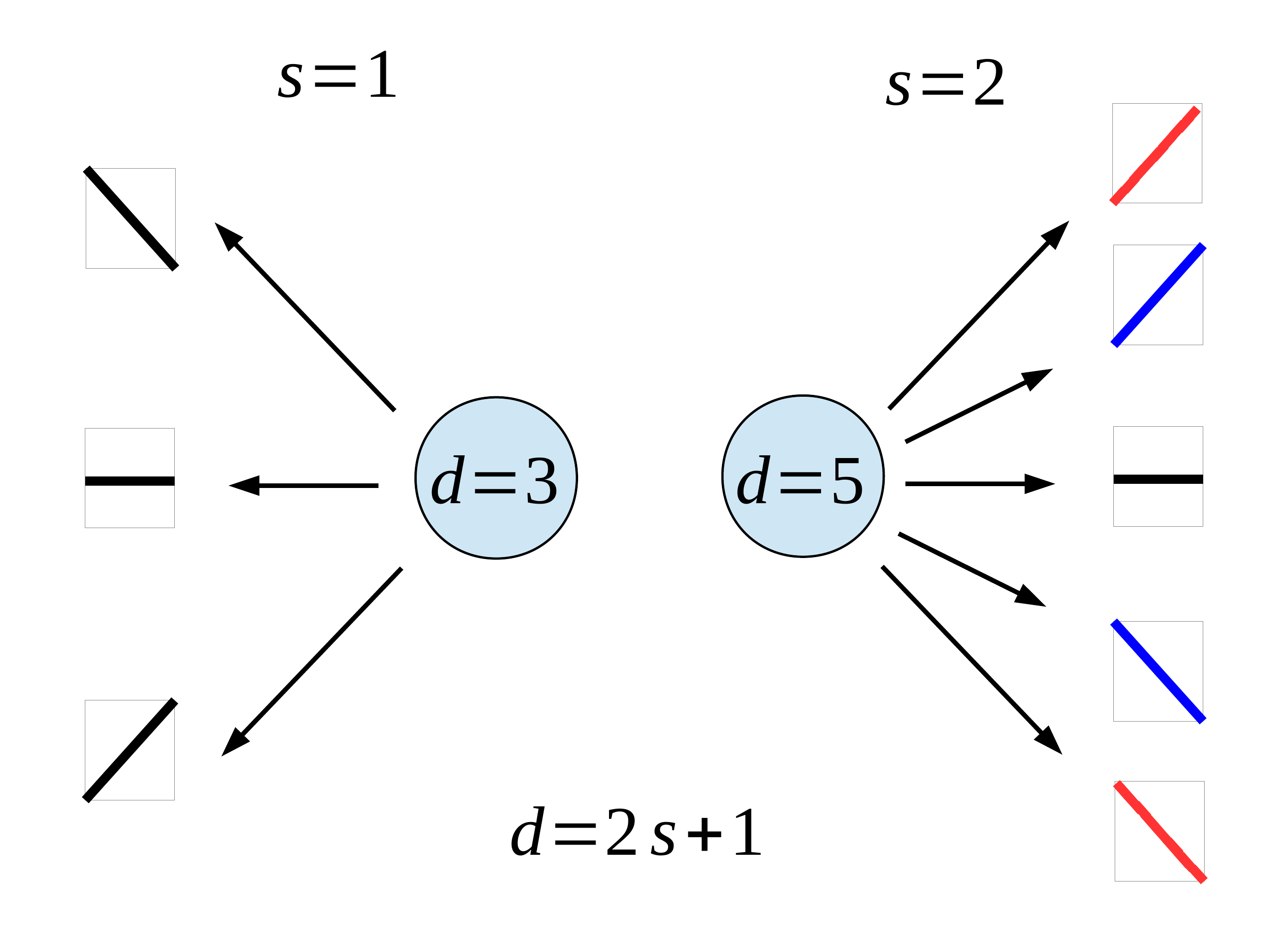}\caption{\label{fig:spin_States}Labels for the $2s+1$ states for $s=1$ and
$s=2$. Note that the flat steps are always black for all $s$. }
\end{figure}
The predecessor of \cite{zhang2016quantum} is the colored Motzkin
spin chain \cite{MovassaghShor2016}, which we now describe. We take
the length of the chain to be $2n$ and consider an integer spin$-s$
chain. As before the $d=2s+1$ spin states are labeled by up and down
steps of $s$ different colors as shown in Fig. \ref{fig:spin_States}.
Equivalently, and for better readability, we instead use the labels
$\left\{u^1,u^2,\cdots,u^s,0,d^1,d^2,\cdots,d^s\right\}$ where $u$
means a step up and $d$ a step down. We distinguish each \textit{type}
of step by associating a color from the $s$ colors shown as superscripts
on $u$ and $d$. Lastly, $0$ denotes a flat step which always has
a single color (black).

A Motzkin walk on $2n$ steps is any walk from $\left(x,y\right)=\left(0,0\right)$
to $\left(x,y\right)=\left(2n,0\right)$ with steps $\left(1,0\right)$,
$\left(1,1\right)$ and $\left(1,-1\right)$ that never passes below
the $x-$axis, i.e., $y\ge0$. An example of such a walk is shown
in Fig. \ref{fig:Motzkin}. Each up step in a Motzkin walk has a corresponding
down step. When the Motzkin walk is $s-$colored, the up steps are
colored arbitrary from $1,\dots,s$, and each down step has the same
color as its corresponding up step (see Fig. \ref{fig:Motzkin}).

Denote by $\Omega_{j}$ the set of all $s-$colored Motzkin paths
of length $j$. Below for notational convenience we simply write $\Omega\equiv\Omega_{2n}$.
In this model the unique ground state is the $s-$colored \textit{Motzkin
state} which is defined to be the uniform superposition of all $s$
colorings of Motzkin walks on $2n$ steps
\[
|{\cal M}_{2n}\rangle=\frac{1}{\sqrt{|{\cal M}_{2n}|}}\sum_{x\in\Omega}|x\rangle.
\]

The half-chain entanglement entropy is asymptotically \cite{MovassaghShor2016}
\begin{eqnarray*}
S & = & 2\log_{2}\left(s\right)\mbox{ }\sqrt{\frac{2\sigma}{\pi}}\mbox{ }\sqrt{n}+\frac{1}{2}\log_{2}\left(2\pi\sigma n\right)+\left(\gamma-\frac{1}{2}\right)\log_{2}e\quad\mbox{bits}
\end{eqnarray*}
where $\sigma=\frac{\sqrt{s}}{2\sqrt{s}+1}$ and $\gamma$ is Euler's
constant. The Motzkin state is a pure state, whose entanglement entropy
is zero. However, the entanglement entropy quantifies the amount of
disorder produced (i.e., information lost) by ignoring a subset of
the chain. The leading order $\sqrt{n}$ scaling of the half-chain
entropy establishes that there is a large amount of quantum correlations
between the two halves. 
\begin{figure}
\begin{centering}
\includegraphics[scale=0.35]{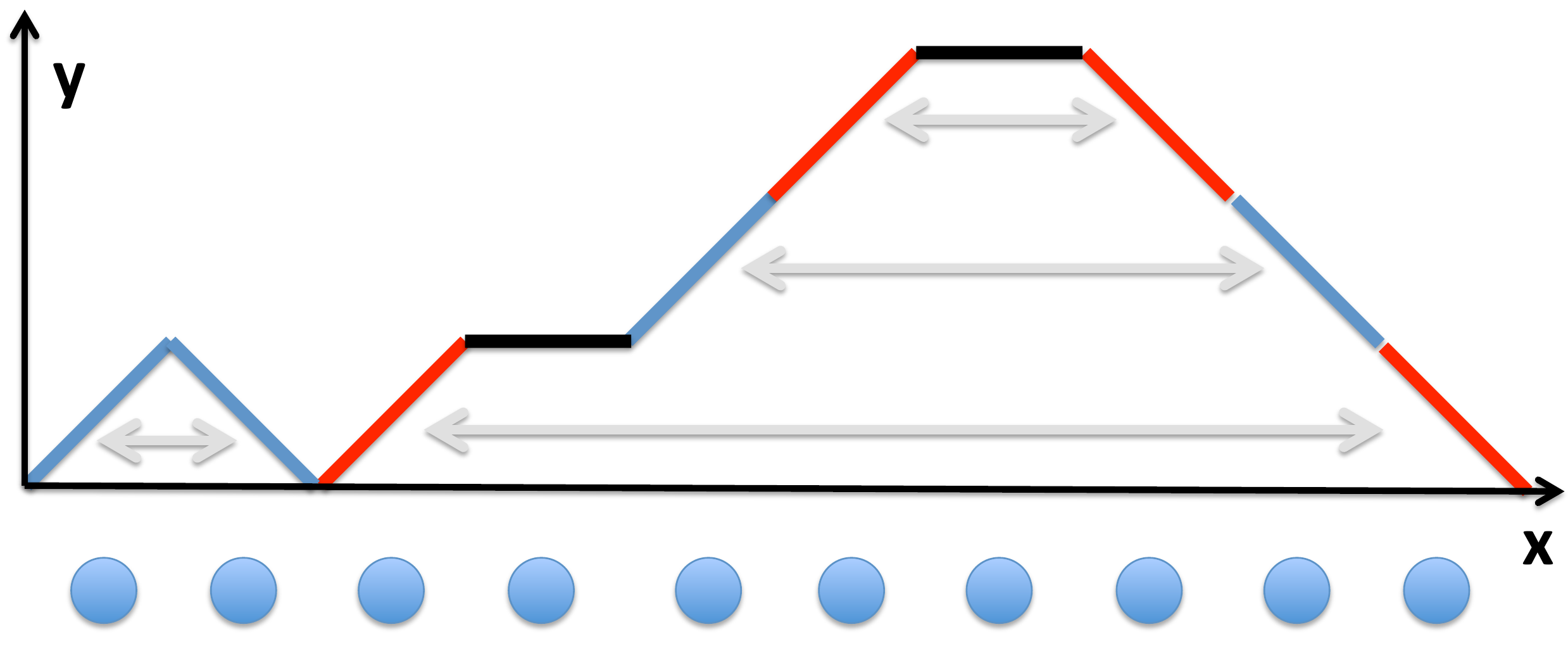}
\par\end{centering}
\centering{}\caption{\label{fig:Motzkin}A Motzkin walk with $s=2$ colors on a chain of
length $2n=10$. }
\end{figure}

Part of the reason that there is $\sqrt{n}$ half-chain entanglement
entropy in \cite{MovassaghShor2016} is that the color of each down
step must match the color of its corresponding up step, and order
$\sqrt{n}$ of these matched pairs are in opposite halves of the chain.
The latter is because, the expected height in the middle of the chain
scales as $\sqrt{n}$, which is a consequence of universality of Brownian
motion and the convergence of Motzkin walks to Brownian excursions. 

Consider the following local operations to any Motzkin walk: interchanging
zero with a non-flat step (i.e., $0u^{k}\leftrightarrow u^{k}0$ or
$0d^{k}\leftrightarrow d^{k}0$) or interchanging a consecutive pair
of zeros with a peak of a given color (i.e., $00\leftrightarrow u^{k}d^{k}$).
These are shown in Fig. \ref{fig:Local-moves-Motzkin}. Any $s-$colored
Motzkin walk can be obtained from another one by a sequence of these
local changes. 

To construct a local Hamiltonian with projectors as interactions that
has the uniform superposition of the Motzkin walks as its zero energy
ground state, each of the local terms of the Hamiltonian has to annihilate
states that are symmetric under these interchanges. Local projectors
as interactions have the advantage of being robust against certain
perturbations \cite{Kraus2008}. This is important from a practical
point of view and experimental realizations.
\begin{figure}
\centering{}\includegraphics[scale=0.45]{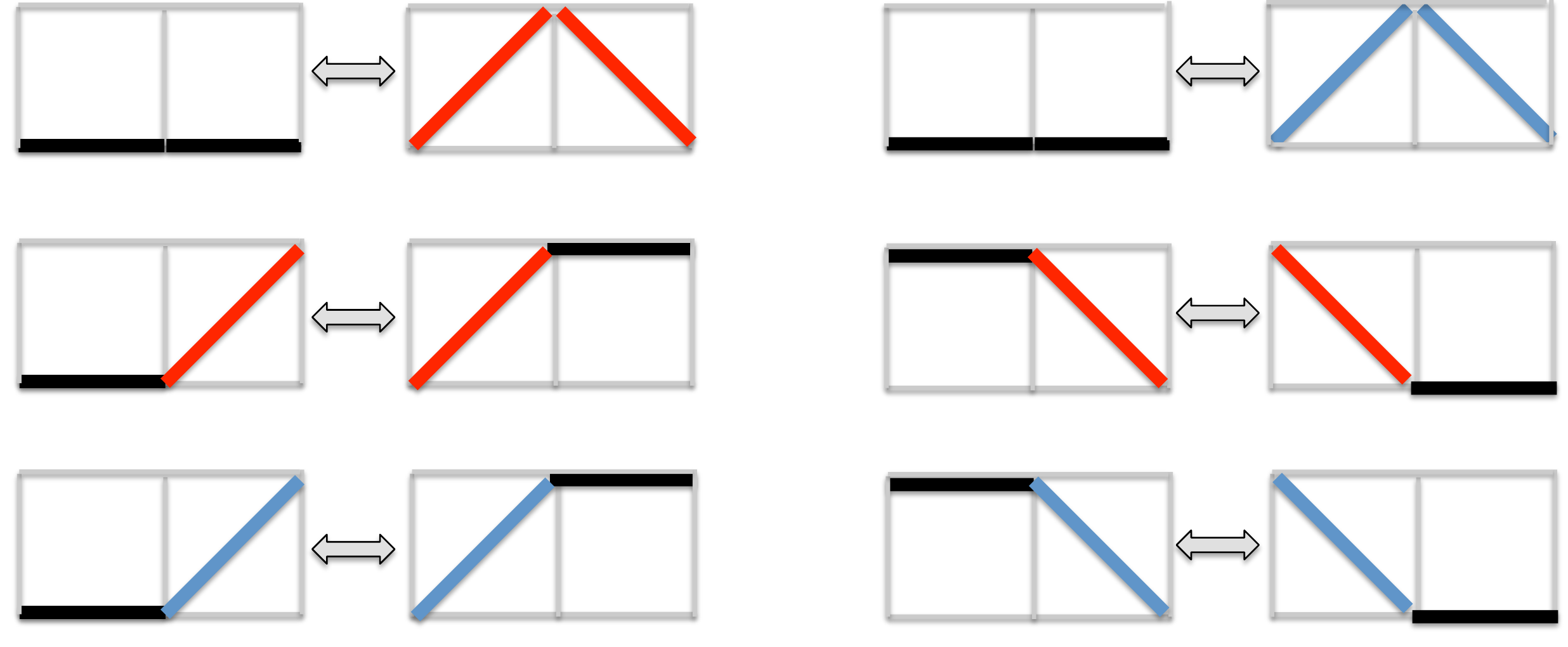}\caption{\label{fig:Local-moves-Motzkin}Local moves for $s=2$.}
\end{figure}

The local Hamiltonian that has the Motzkin state as its unique zero
energy ground state is \cite{MovassaghShor2016}
\begin{equation}
H=\Pi_{boundary}+\sum_{j=1}^{2n-1}\Pi_{j,j+1}+\sum_{j=1}^{2n-1}\Pi_{j,j+1}^{cross},\label{eq:H}
\end{equation}
where $\Pi_{j,j+1}$ implements the local operations discussed above
and is defined by 
\[
\Pi_{j,j+1}\equiv\sum_{k=1}^{s}\left[|U^{k}\rangle_{j,j+1}\langle U{}^{k}|+|D^{k}\rangle_{j,j+1}\langle D^{k}|+|\varphi^{k}\rangle_{j,j+1}\langle\varphi^{k}|\right]
\]
with $|U^{k}\rangle=\frac{1}{\sqrt{2}}\left[|0u^{k}\rangle-|u^{k}0\rangle\right]$,
$|D^{k}\rangle=\frac{1}{\sqrt{2}}\left[|0d^{k}\rangle-|d^{k}0\rangle\right]$
and $|\varphi^{k}\rangle=\frac{1}{\sqrt{2}}\left[|00\rangle-|u^{k}d^{k}\rangle\right]$.
The projectors $\Pi_{boundary}\equiv\sum_{k=1}^{s}\left[|d^{k}\rangle_{1}\langle d^{k}|+|u^{k}\rangle_{2n}\langle u^{k}|\right]$
select out the Motzkin state by excluding all walks that start and
end at non-zero heights. Lastly, $\Pi_{j,j+1}^{cross}=\sum_{k\ne i}|u^{k}d^{i}\rangle_{j,j+1}\langle u^{k}d^{i}|$
ensures that balancing is well ordered (i.e., prohibits $00\leftrightarrow u^{k}d^{i}$);
these projectors are required only when $s>1$ and do not appear in
\cite{Movassagh2012_brackets}. 

\subsection{The weighted Motzkin quantum spin chain}

Let $t>0$ be a real parameter and the state $|GS\rangle$ be defined
by 
\begin{figure}
\centering{}\includegraphics[scale=0.4]{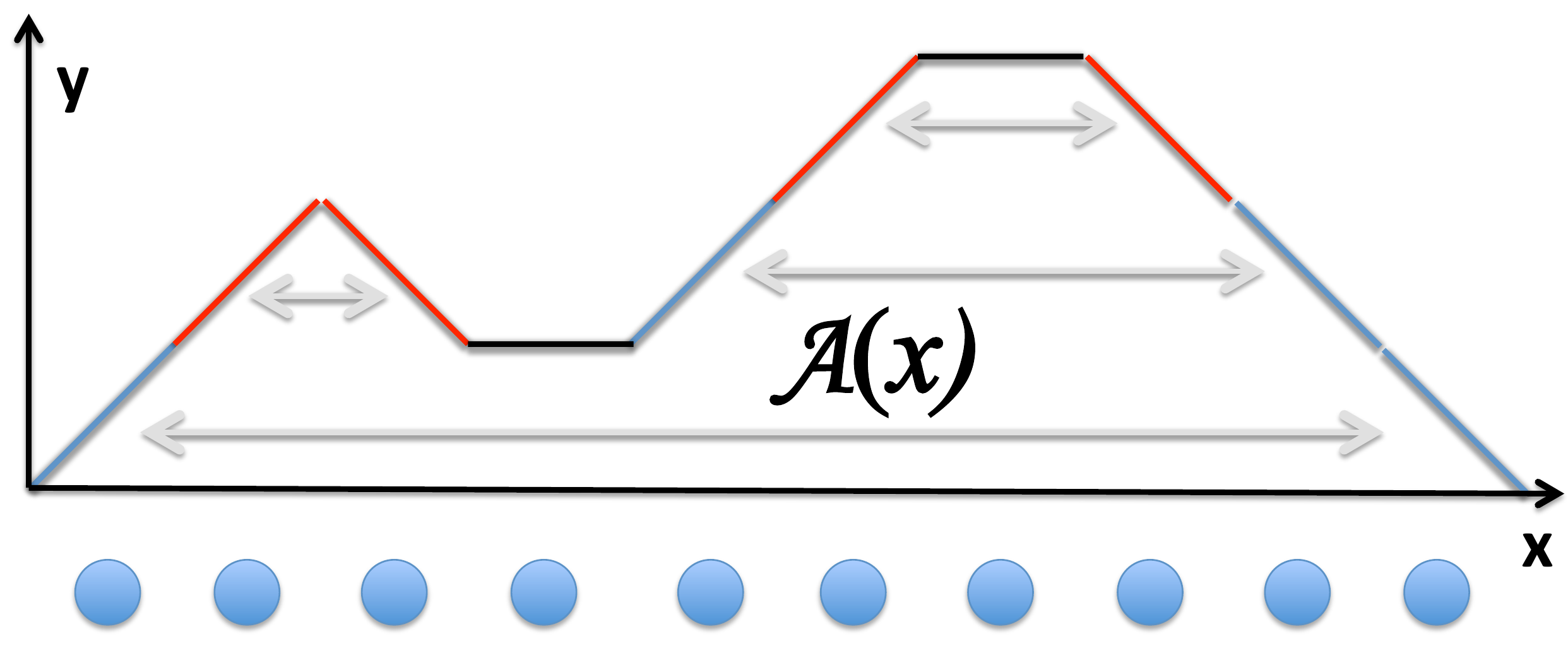}\caption{\label{fig:Weighted-Motzkin-walk}A $2-$colored Motzkin walk, $x$,
with the area ${\cal A}(x)$.}
\end{figure}
\[
|GS\rangle=\frac{1}{\sqrt{Z}}\sum_{x\in\Omega}t^{{\cal A}(x)}|x\rangle,
\]
where ${\cal A}(x)$ is the area enveloped by the Motzkin path and
the $x-$axis (See Fig. \ref{fig:Weighted-Motzkin-walk}), and $Z\equiv\sum_{x\in\Omega}t^{2{\cal A}(x)}$
is the normalization. 

Comment: Taking $t=1$, $|GS\rangle$ becomes the Motzkin state described
above. Taking $t<1$, it was shown in \cite{zhang2016quantum} that
the ground state will have a half-chain entanglement entropy of $O(1)$. 

The local Hamiltonian with the weighted Motzkin state as its unique
zero energy ground state is \cite{zhang2016quantum}
\begin{equation}
H(t)=\Pi_{boundary}+\sum_{j=1}^{2n-1}\Pi_{j,j+1}+\sum_{j=1}^{2n-1}\Pi_{j,j+1}^{cross},\label{eq:H-1}
\end{equation}
where $\Pi_{boundary}$ and $\Pi_{j,j+1}^{cross}$ are as in \cite{MovassaghShor2016}.
However, $\Pi_{j,j+1}$ is deformed and depends on a parameter $t>0$
\begin{equation}
\Pi_{j,j+1}(t)\equiv\sum_{k=1}^{s}\left[|U^{k}(t)\rangle_{j,j+1}\langle U{}^{k}(t)|+|D^{k}(t)\rangle_{j,j+1}\langle D^{k}(t)|+|\varphi^{k}(t)\rangle_{j,j+1}\langle\varphi^{k}(t)|\right]\label{eq:Projectors}
\end{equation}
where $|U^{k}(t)\rangle=\frac{1}{\sqrt{1+t^{2}}}\left[t\text{ }|0u^{k}\rangle-|u^{k}0\rangle\right]$,
$|D^{k}(t)\rangle=\frac{1}{\sqrt{1+t^{2}}}\left[|0d^{k}\rangle-t\text{ }|d^{k}0\rangle\right]$
and $|\varphi^{k}(t)\rangle=\frac{1}{\sqrt{1+t^{2}}}\left[|u^{k}d^{k}\rangle-t\text{ }|00\rangle\right]$.
If $t>1$, then this deformation favors area-increasing moves that
makes the expected height in the middle of the Motzkin walks proportional
to $n$(instead of $\sqrt{n}$). Consequently, the left half of the
chain has about order $n$ step ups whose corresponding down steps
occur on the right half of the chain. The exponential number of possible
colors of the step ups on the left hand side introduces a very large
correlation between the two halves of the chain. This results in a
highly entangled ground state where the entropy $S=\Theta(n)$.

In \cite{zhang2016quantum} they prove that the unique ground state
of $H(t)$ is $|GS\rangle$. They found that the half-chain entanglement
entropy as a function of $t$ and $s$ behaves as 
\[
S_{n}=\left\{ \begin{array}{ccc}
\Theta(n) & \quad & t>1,\text{ }s>1\\
O(1) &  & t<1.
\end{array}\right.
\]
Note that at $t=1$ and $s=1$ the model coincides with \cite{Movassagh2012_brackets}
and has $S_{n}=\Theta(\log(n))$, whereas for $t=1$ and $s>1$, it
coincides with \cite{MovassaghShor2016} that has $S_{n}=\Theta(\sqrt{n})$. 

The intuition behind the $t-$deformation is that one wants to increase
the expected height in the middle of the chain to be $\Theta(n)$.
To do so, they multiply the ket in the projector that favors a local
increase of the area with a parameter $t>1$. This would prefer the
moves that increase the area and consequently push the whole walk
up.

We now prove that the gap of this Hamiltonian for $t>1$ is $\Delta\le8ns\text{ }t^{-n^{2}/3}$.

\subsection{\label{subsec:Klitch}Reduction to the gap to that of an underlying
Markov chain}

There is a general mapping between Stoquastic and frustration free
(FF) local Hamiltonians and classical Markov chains with Glauber dynamics
\cite{bravyi2009complexity}.

Using Perron-Frobenius theorem, Bravyi and Terhal showed that the
ground state $|GS\rangle$ of any stoquastic FF Hamiltonian can be
chosen to be a vector with non-negative amplitudes on a standard basis
\cite{bravyi2009complexity}, i.e., $x\in\Omega$ such that $\langle x|GS\rangle>0$.

Fix an integer $s\geq2$ and a real parameter $t>1$. An $s$-colored
Motzkin walk of length $2n$ is a sequence $x=(x_{1},\ldots,x_{2n})$
with each $x_{j}\in\left\{ u^{1},\dots,u^{s},0,d^{1},\dots,d^{s}\right\} $
starting and ending at height zero with matched colors. 

We now restrict the Hamiltonian to the Motzkin subspace, i.e., span
of $\Omega$ and upper bound the gap in this subspace, which serves
as an upper bound on the gap of the full Hamiltonian. Following \cite[Sec. 2.1]{bravyi2009complexity}
we define the transition matrix that defines a random walk on $\Omega$
\begin{equation}
P(x,y)=\delta_{x,y}-\beta\sqrt{\frac{\pi(y)}{\pi(x)}}\langle x|H|y\rangle\label{eq:P_General-1}
\end{equation}
where $\beta>0$ is real and chosen such that $P(x,y)\ge0$ and $\pi(x)=\langle x|GS\rangle^{2}$
is the stationary distribution. The eigenvalue equation $(\mathbb{I}-\beta H)|GS\rangle=|GS\rangle$
implies that $\sum_{y}P(x,y)=1$. Therefore, the matrix $P$ really
defines a random walk on $\Omega$. 

Below we prove that $P$ is a symmetric Markov chain with a unique
stationary distribution, which means the eigenvalues of $P$ are real
and can be ordered as $1=\lambda_{1}>\lambda_{2}\geq\cdots.$

Recall that the spectral gap of a Markov chain is the difference of
its two largest eigenvalues, i.e. $1-\lambda_{2}(P)$. Also the gap
of $H$ is equal to the gap of the matrix
\[
M(x,y)=\sqrt{\frac{\pi(y)}{\pi(x)}}\langle x|H|y\rangle
\]
because $H$ and $M$ are related by a diagonal similarity transformation.
The minus sign in Eq. \ref{eq:P_General-1} enables us to obtained
the energy gap of the FF Hamiltonian from the spectral gap of the
Markov chain $P$
\begin{equation}
\Delta(H)=\frac{1-\lambda_{2}(P)}{\beta}.\label{eq:GapDefinition-1}
\end{equation}

Let us take $1/\beta=\frac{1}{2ns}\frac{1+t^{2}}{t^{2}}$, whereby
the entries of the matrix $P$ in Eq. \ref{eq:P_General-1} become
\begin{equation}
P(x,y)=\delta_{x,y}-\frac{1}{2ns}\frac{1+t^{2}}{t^{2}}\sqrt{\frac{\pi(y)}{\pi(x)}}\langle x|H|y\rangle.\label{eq:P}
\end{equation}

Therefore, an upper bound on the gap of the Markov chain provides
us with an upper bound on the gap of $H$ via
\begin{equation}
\Delta(H)=\frac{2nst^{2}}{1+t^{2}}(1-\lambda_{2}(P)).\label{eq:GapRelations}
\end{equation}

After some preliminaries, we show that $P$ is a transition matrix
of a reversible Markov chain with the unique stationary distribution 
\begin{equation}
\pi(x)\equiv\frac{t^{2{\cal A}(x)}}{Z},\label{eq:pi_x}
\end{equation}
where $Z=\sum_{x\in\Omega_{2n}}t^{2\mathcal{\mathcal{A}}(x)}$ as
above and obtain an upper bound on $1-\lambda_{2}(P)$. 

Clearly $\sum_{x}\pi(x)=1$ and $0<\pi(x)\le1$. Let us analyze $\langle x|H|y\rangle$,
where $x=x_{1}x_{2}\dots x_{2n}$ is seen as a string with $x_{i}\in\{u^{1},\dots,u^{s},0,d^{1},\dots,d^{s}\}$.
Similarly for the string $y$. 

Since $\Pi_{j,j+1}(t)$ acts locally, it should be clear that $\langle x|H|y\rangle=0$
unless the strings $x$ and $y$ coincide everywhere except from at
most two consecutive positions. That is $\langle x|H|y\rangle\ne0$
if and only if there exists a $j$ such that $\langle x|\Pi_{j,j+1}(t)|y\rangle\ne0$,
which in turn means $x_{1}=y_{1},$$x_{2}=y_{2}$, ... , $x_{j-1}=y_{j-1}$,
$x_{j+2}=y_{j+2}$, $\dots$, $x_{2n}=y_{2n}$. Note that here $x_{j}x_{j+1}$
may or may not be equal to $y_{j}y_{j+1}$. 

In the balanced ground subspace we have 
\begin{equation}
\langle x|H|y\rangle=\langle x|\left\{ \sum_{j=1}^{2n-1}\mathbb{I}_{d^{j-1}}\otimes\Pi_{j,j+1}(t)\otimes\mathbb{I}_{d^{2n-j-1}}\right\} |y\rangle=\sum_{j=1}^{2n-1}\langle x_{j}x_{j+1}|\Pi_{j,j+1}(t)|y_{j}y_{j+1}\rangle.\label{eq:xHy}
\end{equation}

If $x_{j}x_{j+1}=y_{j}y_{j+1}$ , we have a 'diagonal term', and if
$\langle x|H|y\rangle\ne0$ and $x_{j}x_{j+1}\ne y_{j}y_{j+1}$ we
have a 'nonzero off-diagonal term'. The latter occurs if and only
if $x_{j}x_{j+1}\text{ and }y_{j}y_{j+1}$ are related by a local
move. 

Let us look at the largest diagonal term $\langle x|H|x\rangle$,
where $x=y=00\cdots0$. We have 
\[
\sum_{j=1}^{2n-1}\langle x_{j}x_{j+1}|\Pi_{j,j+1}(t)|y_{j}y_{j+1}\rangle=\sum_{j=1}^{2n-1}\langle x_{j}x_{j+1}|\left\{ \sum_{k=1}^{s}|\varphi^{k}(t)\rangle_{j,j+1}\langle\varphi^{k}(t)|\right\} |x_{j}x_{j+1}\rangle,
\]
because only $|\varphi(t)\rangle$ has a $|00\rangle$ in it. The
other projectors vanish. Now Eq. \ref{eq:xHy} becomes
\begin{eqnarray}
|\varphi^{k}(t)\rangle_{j,j+1}\langle\varphi^{k}(t)| & = & \frac{1}{1+t^{2}}\left[|u^{k}d^{k}\rangle-t\text{ }|00\rangle\right]\left[\langle u^{k}d^{k}|-t\text{ }\langle00|\right]\nonumber \\
 & = & \frac{1}{1+t^{2}}\left[|u^{k}d^{k}\rangle\langle u^{k}d^{k}|-t\text{ }|00\rangle\langle u^{k}d^{k}|-t|u^{k}d^{k}\rangle\langle00|+t^{2}\text{ }|00\rangle\langle00|\right].\label{eq:expanded}
\end{eqnarray}
So we have  $\langle00|\left\{ |\varphi^{k}(t)\rangle_{j,j+1}\langle\varphi^{k}(t)|\right\} |00\rangle=\frac{t^{2}}{1+t^{2}}$
and 
\[
\langle00\dots0|\sum_{j=1}^{2n-1}\left\{ \sum_{k=1}^{s}|\varphi^{k}(t)\rangle_{j,j+1}\langle\varphi^{k}(t)|\right\} |00\dots0\rangle=\frac{(2n-1)st^{2}}{1+t^{2}}.
\]
The smallest non-zero diagonal term in $\langle x|H|x\rangle$ corresponds
to a tent shape, where for some $1\le k\le s$, $x_{j}=u^{k}$ for
$j<n$ and $x_{j}=d^{k}$ for $j>n+1$, and $x_{n}x_{n+1}=u^{k}d^{k}$.
In this case, Eq. \ref{eq:xHy} reads
\[
\sum_{j=1}^{2n-1}\langle x|\left\{ \sum_{k=1}^{s}|\varphi^{k}(t)\rangle_{j,j+1}\langle\varphi^{k}(t)|\right\} |x\rangle=\langle u_{n}^{k}d_{n+1}^{k}|\varphi^{k}(t)\rangle_{n,n+1}\langle\varphi^{k}(t)|u_{n}^{k}d_{n+1}^{k}\rangle=\frac{1}{1+t^{2}}.
\]
Based on this analysis and using the definition of $P$ given by Eq.
\ref{eq:P} we conclude
\[
0<P(x,x)\le1-\frac{1}{2nst^{2}}.
\]

Next consider the off diagonal terms. Suppose $x$ and $y$ coincide
everywhere except at $j,j+1$ position. In this case $\langle x|H|y\rangle\ne0$
iff $x_{j}x_{j+1}$ and $y_{j}y_{j+1}$ are related by a local move.
For example suppose $x_{j}x_{j+1}=00$ and $y_{j}y_{j+1}=u^{2}d^{2}$,
then the only term that survives in Eq. \ref{eq:expanded} is $-t\text{ }|00\rangle\langle u^{2}d^{2}|$
in the Hamiltonian. So we have 
\[
\langle00|\left\{ |\varphi^{k}(t)\rangle_{j,j+1}\langle\varphi^{k}(t)|\right\} |u^{2}d^{2}\rangle=-\frac{t}{1+t^{2}}\langle00\text{ }|00\rangle\langle u^{2}d^{2}|u^{2}d^{2}\rangle=-\frac{t}{1+t^{2}}.
\]
Since any local move either increases or decreases the area under
the walk by exactly one, in Eq. \ref{eq:pi_x} we have 
\[
\sqrt{\frac{\pi(y)}{\pi(x)}}=\left\{ \begin{array}{ccc}
t &  & x\rightarrow y\text{ by a local move that increases the area}\\
1/t &  & x\rightarrow y\text{ by a local move that decreases the area}
\end{array}\right.
\]

Therefore, in Eq. \ref{eq:P} the nonzero off diagonal terms in which
$x\ne y$ are 
\[
P(x,y)=-\frac{1}{2ns}\frac{1+t^{2}}{t^{2}}\sqrt{\frac{\pi(y)}{\pi(x)}}\langle x|H|y\rangle
\]
 and the bounds for such terms are
\[
0\le P(x,y)\le\frac{1}{2nst^{2}}.
\]

$P$ is a stochastic matrix as the following shows
\[
\sum_{y}P(x,y)=1-\frac{1}{(2n-1)s}\frac{1+t^{2}}{t^{2}}\frac{1}{\sqrt{\pi(x)}}\langle x|H\sum_{y}\frac{t^{{\cal A}(y)}}{\sqrt{Z}}|y\rangle=1,
\]
where we used $|GS\rangle=\sum_{y}\frac{t^{{\cal A}(y)}}{\sqrt{Z}}|y\rangle$
and the fact that $H|GS\rangle=0$. 

It is easy to see that $P(x,y)$ satisfies detailed balance. Indeed,
since $\langle x|H|y\rangle=\langle y|H|x\rangle$, we have $\pi(x)P(x,y)=\pi(y)P(y,x)$.
Therefore, $\pi(x)$ is a stationary distribution. Since any two Motzkin
walks can be connected by a sequence of local moves \cite{zhang2016quantum}
$P$ is irreducible, lastly $P(x,x)>0$, and the chain is aperiodic.
We conclude that $\pi(x)$ is the unique stationary distribution that
the  chain converges to. 
\subsection{Upper bound on the gap}
For any subset of states $A\subset\Omega$ denote $\pi(A)$ by $\pi(A)\equiv\sum_{x\in A}\pi(x).$
To show that $P$ has a small spectral gap, we are going to identify
two large sets ($S,S'$ below) separated by a small bottleneck ($B$
below) and apply Cheeger's inequality (see Eq. \ref{eq:cheeger} below).
\begin{figure}
\begin{centering}
\includegraphics[scale=0.4]{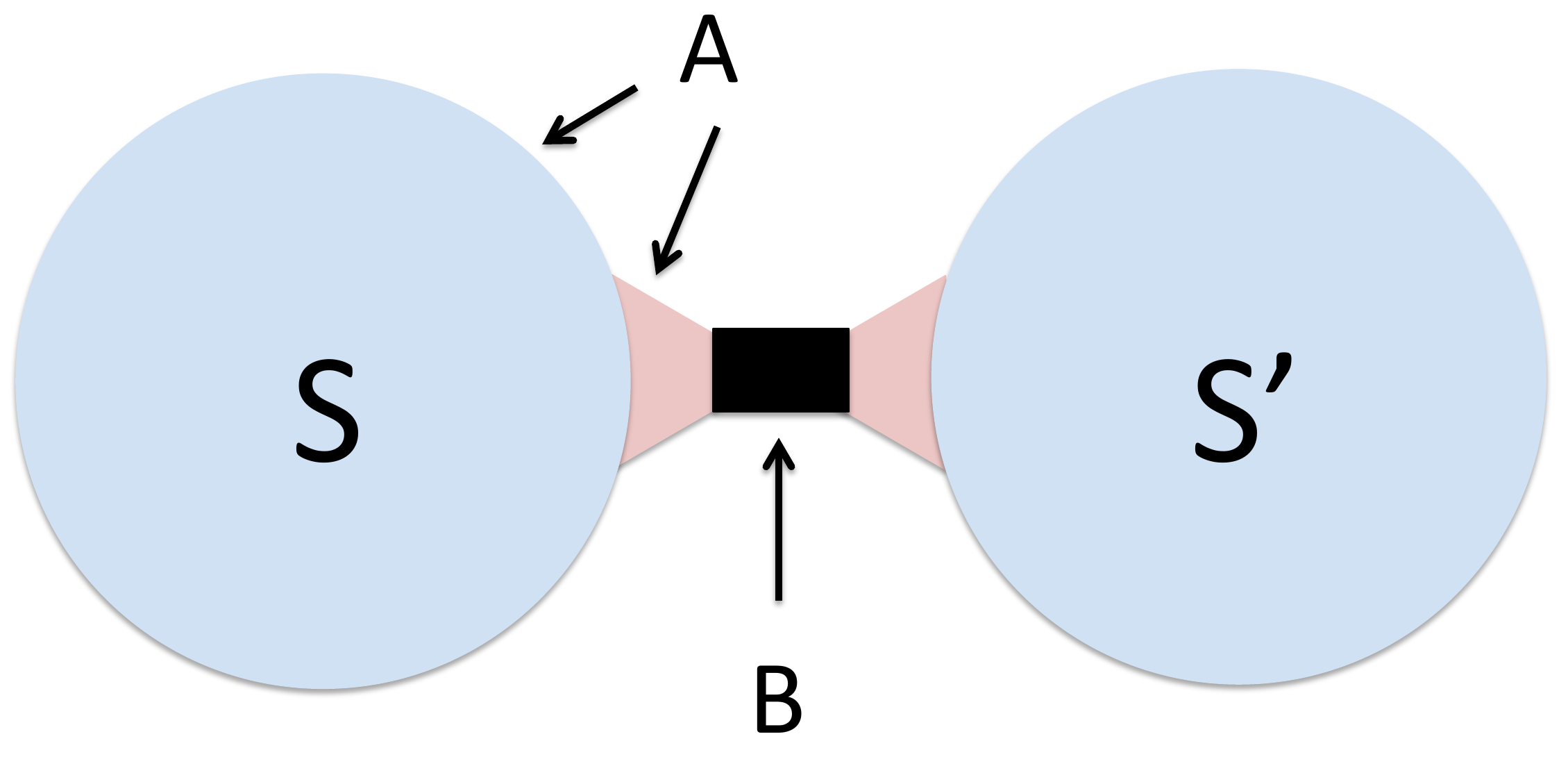}
\par\end{centering}
\caption{The space $\Omega$ of colored Motzkin paths contains two large sets
$S$and $S'$ separated by a small bottleneck $B$. Lemma \ref{lem:big}
below verifies that $S$ and $S'$ are large; Lemma \ref{lem:bottleneck}
that $B$ separates $S$ from $S'$; and Lemma \ref{lem:Bsmall} that
$B$ is small. The proof is completed by applying Cheeger's inequality
to the set $A$ consisting of all paths reachable by a sequence of
local moves starting in $S$ without passing through $B$.}

\end{figure}

\begin{defn*}
For $a\geq0$ let $D_{a}=\left\{ x\in\Omega_{2n}:\mathcal{A}(x)=n^{2}-a\right\} $
be the set of $s$-colored Motzkin walks of length $2n$ with area
defect $a$ whose size we denote by $|D_{a}|$. Let $p(a)$ be the
number of partitions of the integer $a$.
\end{defn*}
\begin{figure}
\includegraphics[scale=0.33]{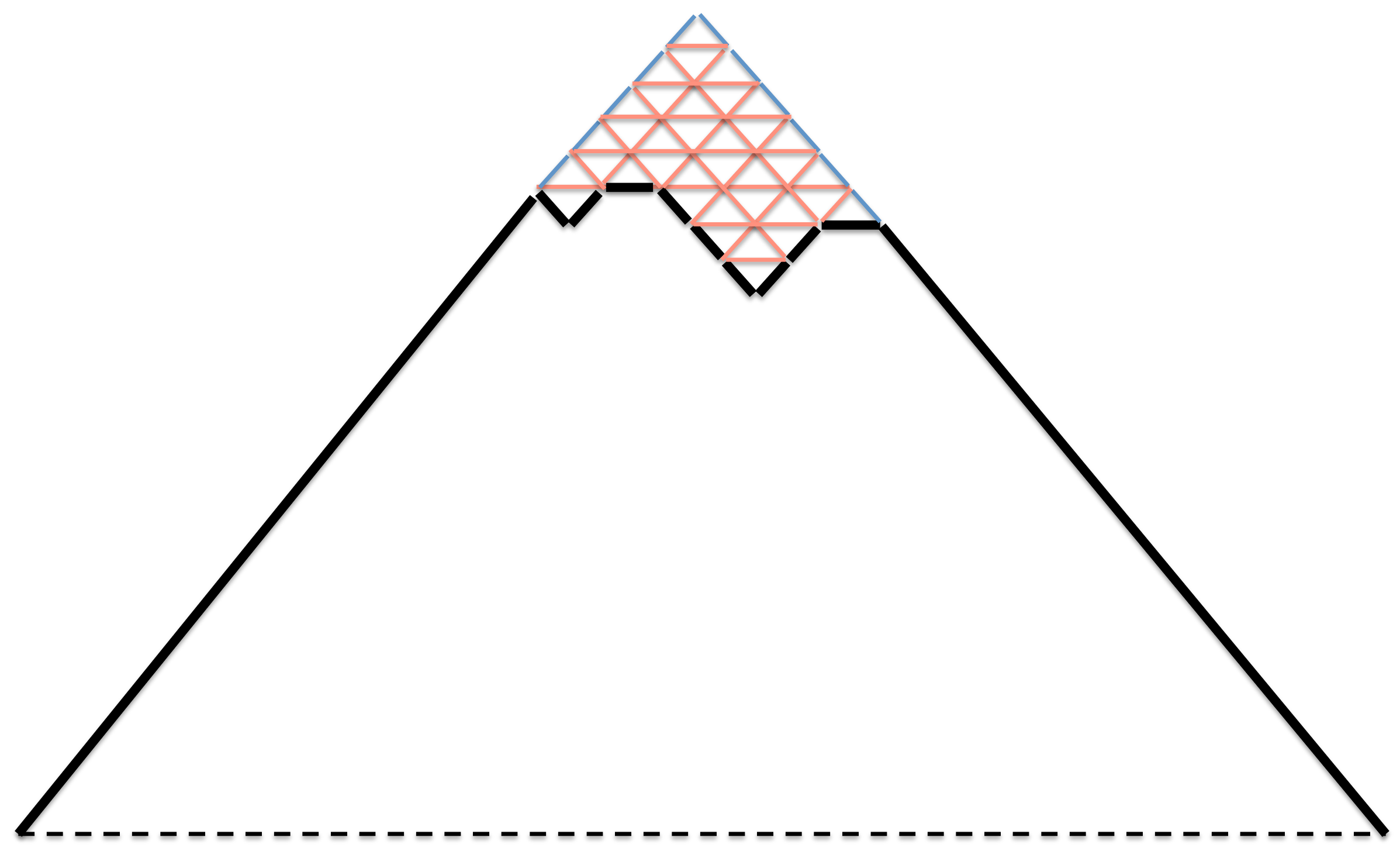}\includegraphics[scale=0.3]{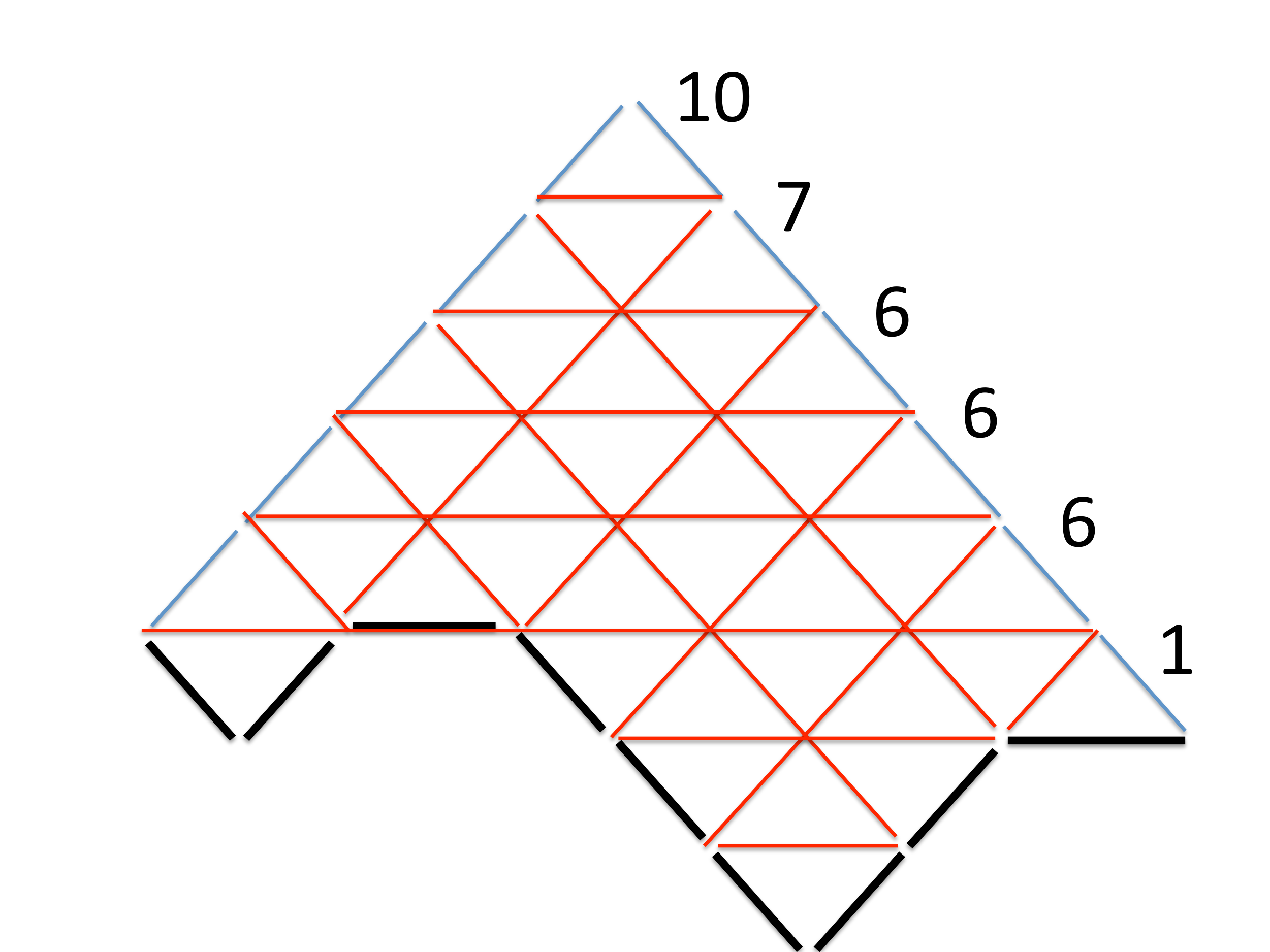}

\caption{\label{fig:Ramanujan}Example of a Motzkin path in $D_{a}$ with area
defect $a=10+7+6+6+6+1$. Its complementary partition is zoomed in
and shown on the right.}

\end{figure}

To verify that $S$ and $S'$ are large and that $B$ is small, we
will use the following lemma, which says that the stationary distribution
$\pi$ concentrates on walks with small area defect. 
\begin{lem}
\label{lem:defect}$\pi(D_{a})<p(a)t^{-2a}<t^{-a}$
\end{lem}
\begin{proof}
The shape of a walk $x\in D_{a}$ is determined by a partition of
$a$, whose triangular diagram lies above $x$ (see Fig. \ref{fig:Ramanujan}).
This diagram is analogous to a Young diagram, with triangles instead
of squares. Since $x$ has at most $n$ up steps, each of which can
be colored arbitrarily in $s$ different ways, we have $|D_{a}|\leq p(a)s^{n}$
where $p(a)$ is the number of partitions of $a$. Using only the
$s^{n}$ tent-shaped paths with full area $n^{2}$ to lower bound
$Z$, we obtain.
\[
\pi(D_{a})=\frac{|D_{a}|t^{2(n^{2}-a)}}{Z_{2n}}<\frac{s^{n}p(a)t^{2(n^{2}-a)}}{s^{n}t^{2n^{2}}}=p(a)t^{-2a}.
\]
The asymptotic form of the number of partitions, $p(a)$, is given
by Hardy-Ramanujan formula, which is sub-exponential with respect
to $a$ \cite{hardy1918asymptotic}:
\[
p(a)\sim\frac{1}{4\sqrt{3}a}\exp\left(\pi\sqrt{2a/3}\right).
\]
In particular, for fixed $t>1$ we have for all sufficiently large
$a$,
\begin{equation}
\pi(D_{a})=\frac{1}{4\sqrt{3}a}\exp\left(-2a\log t+\pi\sqrt{2a/3}\right)<t^{-a}.\label{eq:defect}
\end{equation}
\end{proof}
\begin{defn}
Let $h_{j}(x)$ denote the height of $x$ on the $j^{th}$ step. We
say $x$ is \emph{prime }if $h_{j}(x)>0$ for all $1<j<2n$. Denote
by $\Lambda$ the set of all prime colored Motzkin paths in $\Omega$. 
\end{defn}
In any prime Motzkin walk, the up step in position $1$ must have
the same color as the down step in position $2n$. Hence $\Lambda=S\cup S'$
where
\[
S\equiv\left\{ x\in\Lambda\text{ }:\,color(x_{1})=color(x_{2n})\le s/2\right\} ,
\]
\[
S'\equiv\left\{ x\in\Lambda\text{ }:\,color(x_{1})=color(x_{2n})>s/2\right\} .
\]
Comment: Possibility of $s$ being odd will not hurt the arguments
below. 

We now show that nonprime paths have small areas. Let $\Omega_{j}\times\Omega_{2n-j}$
for the set of concatenations $\tilde{x}\tilde{y}$ where $\tilde{x}\in\Omega_{j}$
and $\tilde{y}\in\Omega_{2n-j}$. Since $\mathcal{A}(\tilde{x}\tilde{y})=\mathcal{A}(\tilde{x})+\mathcal{A}(\tilde{y})$
and $\max_{\tilde{x}\in\Omega_{j}}\mathcal{A}(\tilde{x})=j^{2}/4$,
we have
\[
\max_{x\in\Omega_{j}\times\Omega_{2n-j}}\mathcal{{A}}(x)\leq\frac{j^{2}+(2n-j)^{2}}{4}.
\]

\begin{lem}
\label{lem:big}For any fixed $s\geq2$ and $t>1$, we have for all
sufficiently large $n$,
\[
\pi(S')\geq\pi(S)>\frac{1}{4}.
\]
\end{lem}
\begin{proof}
The first inequality follows from the fact that $S'$ includes at
least as many colors as $S$. By symmetry, $\pi\left\{ x\in\Lambda\text{ }:\,color(x_{1})=color(x_{2n})=i\right\} $
does not depend on $i\in\left\{ 1,\ldots,s\right\} $, and in the
worst case where $s=3$ and $\pi(S)$ is the smallest fraction of
$\pi(\Lambda)$ we have
\[
\pi(S')\geq\pi(S)\geq\frac{1}{3}\pi(\Lambda)
\]
If $x\notin\Lambda$ then $x$ must have a step of height zero in
some position $1<j<2n$. Then ${\cal A}(x)\leq n^{2}-2n+2$; the bound
is saturated when $j=2$ and we have a concatenation of two tents
of heights one and $n-1$. By Eq. \ref{eq:defect}, we obtain for
sufficiently large $n$
\begin{equation}
\pi(\Omega-\Lambda)\leq\sum_{a\geq2(n-1)}\pi(D_{a})<\sum_{a\geq2(n-1)}t^{-a}\approx\frac{t^{-2(n-1)}}{1-t}.
\end{equation}
For sufficiently large $n$ the right side is less than $\frac{1}{4}$,
so $\pi(\Lambda)>\frac{3}{4}$ which completes the proof.
\end{proof}
Next we identify the small bottleneck:
\begin{eqnarray*}
B & \equiv & \left\{ \text{colored Motzkin paths with an up step of height zero in position \ensuremath{n}}\right\} \cup\\
 &  & \left\{ \text{colored Motzkin paths with a down step of height zero in position \ensuremath{n+1}}\right\} .
\end{eqnarray*}
\begin{lem}
\label{lem:bottleneck}Any sequence of local moves from S to S' must
pass through B.
\end{lem}
\begin{proof}
Any walk in $S$ starts with a step up of color $i\le s/2$ and is
matched by a down step of color $i$. For the sake of concreteness
call this color $i$ ``blue''. Any walk in $S'$ starts with an
up step and ends with a down step of color $\ell>s/2$, which for
the sake of concreteness we call ``red''. For the walk in $S$ to
go to the walk in $S'$ by a sequence of local moves, the up and down
blue steps would have to meet and annihilate (i.e. become $00$) at
some intermediate position. Then a matched pair of red up and down
steps would have to be created at some intermediate point $1<j<2n$.
Eventually the red step up moves all the way to the left (i.e., position
$1$) via a sequence of local moves and the red step down moves all
the way to the right (i.e., position $2n$). If $j\geq n$ then the
red up step would necessarily pass through position $n$; similarly
if $j<n$ the red down step would pass through the position $n+1$.
\end{proof}
\begin{lem}
\label{lem:Bsmall}For fixed $t>1$ we have $\pi(B)<t^{-n^{2}/3}$
for all sufficiently large $n.$
\end{lem}
\begin{proof}
There is a quadratic area cost to being in $B$: Namely, for any $x\in B$
we have for $n>1$
\[
\mathcal{A}(x)\leq\frac{(n-2)^{2}+(n+2)^{2}}{4}<\frac{1}{2}n^{2}+2n.
\]
Hence by Eq. \ref{eq:defect} we have for sufficiently large $n$
\begin{equation}
\pi(B)<\sum_{a\geq\frac{1}{2}n^{2}-2n}\pi(D_{a})<\sum_{a\ge\frac{1}{2}n^{2}-2n}t^{-a}<\frac{t^{-\frac{1}{2}n^{2}+2n}}{1-t}<t^{-n^{2}/3}.
\end{equation}
\end{proof}
\begin{thm}
For fixed $s\geq2$ and $t>1$ we have $1-\lambda_{2}(P)<8nst^{-n^{2}/3}$,
and $\Delta(H)<8ns\text{ }t^{-n^{2}/3}$ for all sufficiently large
$n.$
\end{thm}
\begin{proof}
We are going to use (the easy part of) Cheeger's inequality for reversible
Markov chains, which says that for any $A\subset\Omega$ with $\pi(A)\leq\frac{1}{2}$
we have (See \cite[Theorem 13.14 and Eqs (7.4)-(7.6)]{levin2009markov})
\begin{equation}
1-\lambda_{2}\leq\frac{2Q(A,A^{c})}{\pi(A)},\label{eq:cheeger}
\end{equation}
where denoting by $Q(x,y)\equiv\pi(x)P(x,y)=\pi(y)P(y,x)$, we have
$Q(A,A^{c})\equiv\underset{x\in A,y\in A^{c}}{\sum}Q(x,y)$. 

Our choice of $A$ is 
\[
A\equiv\{x\in\Omega|\text{ there exists a sequence of legal moves }x_{0},\ldots,x_{k}=x\text{ with }x_{0}\in S\text{ and all }x_{i}\not\in B\}.
\]
By definition, every edge from $A$ to $A^{c}$ is an edge from $A$
to $B$, so by Lemma \ref{lem:Bsmall} and sufficiently large $n$
we have
\[
Q(A,A^{c})=\sum_{x\in A,y\in B}Q(x,y)\leq\sum_{x\in\Omega,y\in B}\pi(y)P(y,x)=\pi(B)<t^{-n^{2}/3}
\]
On the other hand, by Lemma \ref{lem:big}
\[
\pi(A)\geq\pi(S)>\frac{1}{4}.
\]

Finally, we need to show that $\pi(A)\leq\frac{1}{2}$ so that Cheeger's
inequality may be applied. To do so, we show that if $x\in A$ then
$x'\in A^{c}$ for any $x'$ that is a Motzkin walk of the same shape
as $x$ whose colors at each step are changed according to $i\rightarrow s-i+1$.
Indeed, since $x\in A$ there is a sequence of legal moves $x_{0},\ldots,x_{k}=x\text{ with }x_{0}\in S\text{ and all }x_{i}\not\in B$.
Noting that $B=B'$, we have that $x',x'_{k-1},\ldots,x'_{0}$ is
a sequence of local moves from $x'$ to $x'_{0}\in S'$ with all $x'_{i}\notin B$.
By Lemma \ref{lem:bottleneck} it follows that any sequence of local
moves from $S$ to $x'$ must pass through $B$. Hence $x'\notin A$.

By Eq. \ref{eq:cheeger} we conclude that $1-\lambda_{2}<8t^{-n^{2}/3}$. This is the desired result because from Eq. \ref{eq:GapRelations}
we have ($t>1$)
\begin{equation}
\Delta(H)=\frac{2nst^{2}}{1+t^{2}}(1-\lambda_{2}(P))<8ns\text{ }t^{-n^{2}/3}.\label{eq:Gap_Final}
\end{equation}
\end{proof}
This shows that so far the only model that violates the area-law by
a factor of $n$ in one-dimension and satisfies physical reasonability
criteria of translational invariance, locality and uniqueness of ground
state, has an exponentially small energy gap. It would be very interesting
if a physically reasonable model could be proposed in which the half-chain
entanglement entropy violates the area law by $n$ and that the gap
would vanish as a power-law in $n$. 
\section{Acknowledgements}
LL was supported by NSF grant \href{http://www.nsf.gov/awardsearch/showAward?AWD_ID=1455272}{DMS-1455272}
and a Sloan Fellowship. RM is grateful for the support and freedom
provided by the Herman Goldstine fellowship in mathematical sciences
at IBM TJ Watson Research Center. RM also thanks the Simons Foundation
and the American Mathematical Society for the AMS-Simons travel grant.

\bibliographystyle{apsrev4-1}
\bibliography{mybib}

\end{document}